\documentclass[journal]{IEEEtran}

\ifCLASSINFOpdf
\else
\usepackage[dvips]{graphicx}
\fi
\usepackage{url}
\usepackage{graphicx}
\usepackage{amsmath,epsfig,amsfonts,amssymb,graphics,psfrag,amsthm,calc,url,bm,cite,subcaption}
\usepackage{caption}
\usepackage{soul}
\usepackage{algorithm,algpseudocode}
\usepackage{mdframed}
\usepackage{verbatim}
\usepackage{empheq}
\usepackage{booktabs}

\newtheorem{lemma}{Lemma}

\newcommand{\trans}{\mathrm{T}}
\newcommand{\herm}{\mathrm{H}}

\definecolor{orange}{RGB}{255,107,0}

\begin{document}
	\title{Symbol-Level Precoding for Continuous-Aperture ISAC Systems}
	
	\author{\IEEEauthorblockN{Hongli Liu and Qiang Li}
		\thanks{H. Liu and Q. Li are with School of Information and Communication  Engineering, University of Electronic Science and Technology of China, Chengdu, P.~R.~China, 611731. Q. Li is also affliated with Tianfu Jiangxi Laboratory, Chengdu 641419, China. This work was supported by the National Natural Science Foundation of China under Grant 62171110. ({\it Corresponding author: Q. Li, E-mail: lq@uestc.edu.cn.})
		}
		
	}
	
	\maketitle

	\begin{abstract}
		Continuous-aperture arrays (CAPAs) offer rich electromagnetic degrees of freedom for integrated sensing and communication (ISAC), but optimizing continuous current distributions leads to challenging infinite-dimensional problems. This paper investigates a CAPA-enabled downlink ISAC system with symbol-level precoding and receive polarization combining. The transmit current is optimized to maximize weighted target-illumination power while enforcing constructive-interference constraints for communication users. To address the resulting infinite-dimensional nonconvex problem, we establish that the optimal current distribution lies in a finite-dimensional subspace spanned by the communication and sensing electromagnetic responses. This result yields an exact, structure-preserving reformulation in terms of finite-dimensional coefficients. A penalty projected-gradient algorithm is then developed to jointly optimize the current coefficients and polarization combiners. Simulation results demonstrate that the proposed framework achieves higher sensing utility and improved communication reliability than conventional Fourier-basis CAPA and spatially discrete array baselines.
	\end{abstract}

	\IEEEpeerreviewmaketitle

\section{Introduction}
\IEEEPARstart{I}{ntegrated} sensing and communication (ISAC) has emerged as a key paradigm for future wireless networks, enabling communication and radar sensing to share spectrum, hardware, and transmitted waveforms. Conventional ISAC systems rely on multiple-input multiple-output (MIMO) arrays to provide spatial degrees of freedom for joint communication and sensing beamforming. Recent studies have further explored reconfigurable-intelligent-surface (RIS)-enhanced systems~\cite{yang2025secure} and movable-antenna-based designs~\cite{yang2026movable,wang2026movable,xiu2025movable} to reshape wireless channels through passive reflection, antenna-position reconfiguration, or platform-dependent beamforming. Different from these approaches, continuous-aperture arrays (CAPAs) approach the continuous limit of dense antenna arrays and directly control the electromagnetic current distribution over the transmit aperture. By replacing discrete antenna weights with continuous aperture excitation, CAPAs can exploit richer electromagnetic degrees of freedom in field distribution and polarization, enabling more flexible ISAC waveform design.

CAPAs have recently attracted growing attention in both communication and sensing. For communication, finite-basis pattern design has been investigated for multi-user CAPA-MIMO systems~\cite{zhang2023pattern}, while low-complexity continuous beamforming has been developed beyond conventional Fourier-basis representations~\cite{wang2025beamforming}. For sensing, vector-valued continuous aperture excitation has enabled tri-polarized direction and attitude estimation~\cite{si2025joint} and Cram\'er--Rao-bound-oriented source-current optimization for near-field sensing~\cite{jiang2025cramer}. These studies show that CAPAs are not merely high-resolution antenna arrays, but continuous electromagnetic platforms capable of flexibly shaping information-bearing and sensing-oriented fields.

The above potential naturally motivates CAPA-enabled ISAC. Recent works have optimized continuous beamforming to characterize communication--sensing tradeoffs~\cite{zhang2025exploiting} and analyzed CAPA-based ISAC over fading channels using continuous-operator models~\cite{zhao2026continuous}. However, these studies mainly focus on rate-oriented performance analysis or block-level beamforming. The instantaneous symbol structure of communication signals, which provides additional waveform-level degrees of freedom for ISAC, remains largely unexplored in CAPA electromagnetic design.

Symbol-level precoding (SLP) offers a natural way to exploit such instantaneous symbol information. Existing continuous-aperture waveform optimization has been used to enhance sensing beampatterns while suppressing multiuser interference~\cite{ye2026optimal}. In contrast, constructive-interference (CI)-based SLP does not simply suppress interference; instead, it exploits constellation geometry to steer multiuser interference toward regions that facilitate reliable symbol detection~\cite{masouros2015exploiting}. This feature makes CI-based SLP particularly appealing for CAPA-enabled ISAC, where the continuous aperture can be optimized at the symbol level to support both communication detection and sensing illumination.

Motivated by this, we develop an SLP framework for a CAPA-enabled downlink ISAC system with continuous vector-current excitation and receive polarization combining. The transmit current is optimized at the symbol level to satisfy the CI requirements of communication users while maximizing the weighted illumination power toward multiple sensing targets. This leads to an infinite-dimensional nonconvex functional optimization problem coupled with polarization-combiner design. By exploiting the electromagnetic-response structure of the CAPA system, we show that an optimal current can be exactly represented in a finite-dimensional subspace spanned by the communication and sensing responses. The resulting structure-preserving reformulation is solved via a penalty projected-gradient algorithm. Numerical results demonstrate that the proposed response-subspace CAPA design outperforms conventional Fourier-basis CAPA and spatially discrete array (SPDA) benchmarks in both sensing illumination and communication reliability.

\section{System Model and Problem Formulation}
\label{sec:system_model}
We consider a downlink CAPA-aided ISAC system, as shown in Fig.~\ref{fig:system}, where a base station (BS) equipped with a continuous planar aperture simultaneously serves $K$ communication users and illuminates $Q$ sensing targets. The BS controls the continuous current distribution over the aperture, and each communication user employs a tri-polarized receiver to combine the received electric-field components.	
\begin{figure}[t]
	\centering
	\includegraphics[width=0.7\linewidth]{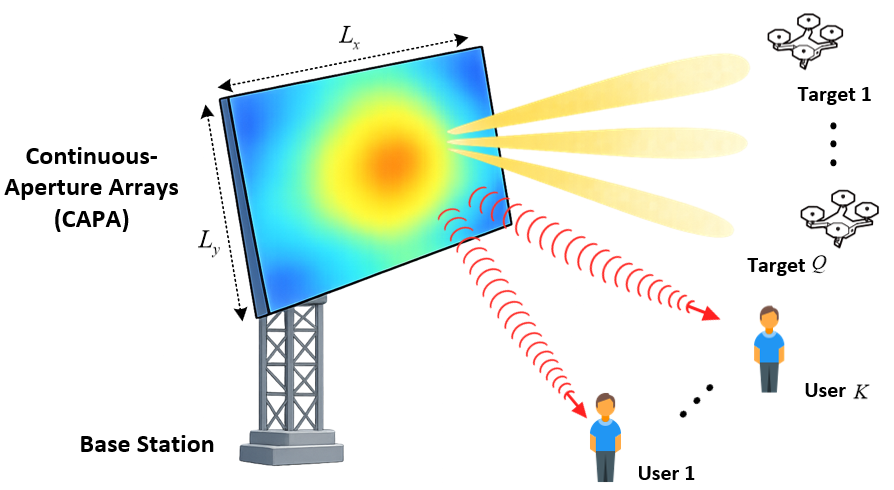}
	\caption{System model of CAPA--aided ISAC.}
	\label{fig:system}
\end{figure}
\vspace{-10pt}
\subsection{Continuous-Aperture Transmit Model}
\label{subsec:transmit_model}
The transmit aperture is modeled as a rectangular continuous surface
located on the \(x\)-\(y\) plane, with side lengths \(L_x\) and \(L_y\),
and centered at the origin, i.e.
	\begin{equation}
		\mathcal S_T
		=
		\left\{
		\mathbf s=[s_x,\, s_y,\, 0]^\trans:
		|s_x| \leq \frac{L_x}{2},~
		|s_y| \leq \frac{L_y}{2}
		\right\},
		\label{eq:aperture_region}
	\end{equation}
where $L_x$ and $L_y$ denote the aperture lengths along the $x$- and $y$-axes, respectively.
	
At symbol interval \(t\), the transmit excitation over the aperture is represented by the equivalent vector current distribution
	\begin{equation}
		\mathbf j_t(\mathbf s)
		=[j_{x,t}(\mathbf s), \ j_{y,t}(\mathbf s), \ j_{z,t}(\mathbf s)]^\trans
		\in\mathbb C^3,
		\quad
		\mathbf s\in\mathcal S_T,
		\label{eq:current_distribution}
	\end{equation}
where the three entries describe the controllable vector-current components of the tri-polarized continuous-aperture transmitter. For a block of $T$ symbol intervals, the average transmit-power constraint is given by
	\begin{equation}
		\frac{1}{T}
		\sum_{t=1}^{T}
		\int_{\mathcal S_T}
		\|\mathbf j_t(\mathbf s)\|_2^2
		d\mathbf s
		\le
		P_{\max}.
		\label{eq:power_constraint}
	\end{equation}

\subsection{Communication Model and CI Constraints}
\label{subsec:comm_metric}
	
We assume that the communication users are located in the far-field region of the CAPA.
Following the far-field approximation of the dyadic Green's function~\cite{ye2026optimal}, the electromagnetic channel kernel from the continuous aperture to user $k$ is given by
\begin{equation}
	\boldsymbol\Gamma_k(\mathbf s)
	=
	\frac{\alpha_0}{R_k}
	e^{-j k_0 R_k}
	e^{j k_0\mathbf q_k^\trans \mathbf s}
	\left(
	\mathbf I_3-\mathbf q_k\mathbf q_k^\trans
	\right)
	\in\mathbb C^{3\times 3},
	\label{eq:comm_kernel}
\end{equation}
where $k_0=2\pi/\lambda$, $R_k$ denotes the user distance, $\alpha_0$ is an EM wave propagation-scaling constant, and $\mathbf q_k=\mathbf q(\theta_k,\phi_k)$ is the unit direction of user $k$. The function $\mathbf q(\theta,\phi)$ is defined as
\begin{equation}
	\mathbf q(\theta,\phi)
	=
	[\cos\theta\sin\phi,~\sin\theta\sin\phi,~\cos\phi]^\trans,
\end{equation}
where $\theta$ and $\phi$ denote the azimuth and broadside angles, respectively.
The received electric field at user $k$ is 
	\begin{equation}
		\mathbf e_k[t]
		=
		\int_{\mathcal S_T}
		\boldsymbol\Gamma_k(\mathbf s)
		\mathbf j_t(\mathbf s)
		d\mathbf s
		\in\mathbb C^3 .
		\label{eq:user_field}
	\end{equation}
Each user employs a tri-polarized receiver to combine the three electric-field components. The combined received signal is given by
	\begin{equation}
		y_k[t]
		=
		\boldsymbol\psi_k^H\mathbf e_k[t]
		+
		n_k[t],
		\label{eq:received_signal}
	\end{equation}
where $\boldsymbol\psi_k\in\mathbb C^3$ denotes the receive polarization combiner of user $k$ with $\|\boldsymbol\psi_k\|_2=1$, and $n_k[t]\sim\mathcal{CN}(0,\sigma_k^2)$ denotes the effective receiver noise.
Accordingly, the noiseless received signal is
	\begin{equation}
		\tilde y_k[t]
		=
		\boldsymbol\psi_k^H
		\int_{\mathcal S_T}
		\boldsymbol\Gamma_k(\mathbf s)
		\mathbf j_t(\mathbf s)
		d\mathbf s .
		\label{eq:noiseless_signal}
	\end{equation}
	
We consider $M$-phase-shift-keying (PSK)\footnote{The subsequent derivations and the result of Lemma~\ref{lemma1} extend directly to QAM signaling.} signaling for all communication users, where $s_k[t]=|s_k[t]|e^{j\angle s_k[t]}$ denotes the intended symbol of user $k$ at symbol time $t$. Since PSK detection depends on the signal phase, the noiseless received signal is rotated to the desired symbol direction as
	\begin{equation}
		z_k[t]
		=
		e^{-j\angle s_k[t]}
		\tilde y_k[t].
		\label{eq:rotated_signal}
	\end{equation}	
Let $\Phi=\pi/M$ denote the  half-angle of the \(M\)-PSK decision region. Following the CI principle~\cite{masouros2015exploiting}, the received signal is required to remain inside the desired decision region with a prescribed safety margin \(\beta_k>0\), i.e.,
	\begin{equation}
		\Re\{z_k[t]\}\sin\Phi
		-
		|\Im\{z_k[t]\}|\cos\Phi
		\ge
		\beta_k,
		\quad
		\forall k,t.
		\label{eq:ci_constraint_abs}
	\end{equation}
The parameter $\beta_k$ controls the robustness of user \(k\) against receiver noise and residual interference. 
The CI constraint in \eqref{eq:ci_constraint_abs} can then be cast into a real and linear form:
	\begin{equation}
		\Re\{\eta_m z_k[t]\}
		\ge
		\beta_k,
		\quad
		m=1,2,\ \forall k,t,
		\label{eq:ci_two_boundary}
	\end{equation}
where $\eta_1=\sin\Phi+j\cos\Phi$ and $\eta_2=\sin\Phi-j\cos\Phi$.

\subsection{Sensing Model and Illumination Metric}
\label{subsec:sensing_metric}
We consider a multi-target illumination scenario, where the sensing kernel toward target $q$ is modeled similarly to that of the communication users as
	\begin{equation}
		\boldsymbol\Gamma_q^{\rm s}(\mathbf s)
		=
		e^{j k_0\mathbf q_q^\trans \mathbf s}
		\left(
		\mathbf I_3-\mathbf q_q\mathbf q_q^\trans
		\right)
		\in\mathbb C^{3\times 3},
		\label{eq:sensing_kernel}
	\end{equation}
	where $\mathbf q_q=\mathbf q(\theta_q,\phi_q)$ denotes the unit direction of target $q$. The target-dependent propagation attenuation is omitted here as it is absorbed into the
	corresponding sensing weight introduced below.

	The radiated electric-field toward target $q$ at  time $t$ is
	\begin{equation}
		\mathbf E_q[t]
		=
		\int_{\mathcal S_T}
		\boldsymbol\Gamma_q^{\rm s}(\mathbf s)
		\mathbf j_t(\mathbf s)
		d\mathbf s
		\in\mathbb C^3 .
		\label{eq:sensing_field}
	\end{equation}

To quantify the overall multi-target illumination performance, we define the sensing utility as
\begin{equation}
	U_{\rm sen}
	=
	\sum_{t=1}^{T}
	\sum_{q=1}^{Q}
	\omega_q
	\|\mathbf E_q[t]\|_2^2,
	\label{eq:sensing_utility}
\end{equation}
where $\omega_q\geq0$ specifies the sensing priority of target $q$. The utility in \eqref{eq:sensing_utility} measures the weighted transmit-side illumination energy over the symbol block, and is suitable for illumination-oriented sensing tasks such as target probing, search, and multi-target focusing.

\subsection{Problem Formulation}
\label{subsec:problem_formulation}
Based on the above models, we jointly optimize the continuous current distributions $\{\mathbf j_t(\mathbf s)\}_{t=1}^{T}$ and receive polarization combiners $\{\boldsymbol\psi_k\}_{k=1}^{K}$ to maximize the weighted sensing illumination utility under the CI and power constraints:
	\begin{subequations}
		\label{eq:continuous_problem}
		\begin{align}
	\hspace{-10pt}		\mathcal P_0:\,
			\max_{\{\mathbf j_t(\mathbf s)\},\{\boldsymbol\psi_k\}}
			\quad
			&
			U_{\rm sen}
			\label{eq:continuous_problem_obj}
			\\
			\mathrm{s.t.}\quad
			&
			\Re\{\eta_m z_k[t]\}
			\ge
			\beta_k,
			\,
			m=1,2,\ \forall k,t,
			\label{eq:continuous_problem_ci}
			\\
			&
			\sum_{t=1}^{T}
			\int_{\mathcal S_T}
			\|\mathbf j_t(\mathbf s)\|_2^2
			d\mathbf s
			\le
			TP_{\max},
			\label{eq:continuous_problem_power}
			\\
			&
			\|\boldsymbol\psi_k\|_2=1,
			\quad
			\forall k.
			\label{eq:continuous_problem_combiner}
		\end{align}
	\end{subequations}
Problem $\mathcal P_0$ is an infinite-dimensional nonconvex problem due to the continuous current variables and their coupling with the polarization combiners. We next exploit the CAPA electromagnetic-response structure to obtain an equivalent finite-dimensional reformulation.

\vspace{-10pt}
\section{Subspace-Based Optimization Approach}
\label{sec:representer_approach}
This section develops a response-subspace optimization approach for problem $\mathcal P_0$. We first show that the continuous current can be restricted, without loss of optimality, to a finite-dimensional subspace induced by the communication and sensing kernels. Based on this result, we reformulate $\mathcal P_0$ as a compact coefficient optimization problem and then solve it using a penalty projected-gradient algorithm.
\vspace{-10pt}
	\subsection{Electromagnetic-Response Subspace}
	\label{subsec:channel_subspace}
	
	We first collect the aperture-domain electromagnetic responses that enter the communication and sensing functionals. Specifically, define
	\begin{equation*} 
		\mathbf G(\mathbf s)
		=
		[
		\boldsymbol\Gamma_1^\herm(\mathbf s),
		\ldots,
		\boldsymbol\Gamma_K^\herm(\mathbf s),
		\left(\boldsymbol\Gamma_1^{\rm s}(\mathbf s)\right)^\herm,
		\ldots,
		\left(\boldsymbol\Gamma_Q^{\rm s}(\mathbf s)\right)^\herm
		].
		\label{eq:response_matrix}
	\end{equation*}
	The corresponding electromagnetic-response subspace is then defined as
	\begin{equation}
		\mathcal R
		=
		\left\{
		\mathbf j(\mathbf s):
		\mathbf j(\mathbf s)
		=
		\mathbf G(\mathbf s)\mathbf c,\ 
		\mathbf c\in\mathbb C^{3(K+Q)}
		\right\}.
		\label{eq:channel_response_subspace}
	\end{equation}
	This subspace collects the current components that are observable through the communication and sensing response kernels. Its dimension is
	\begin{equation*}
		D
		\triangleq
		\dim(\mathcal R)
		\le
		3(K+Q),
		\label{eq:D_definition}
	\end{equation*}
	which equals the number of linearly independent aperture-domain response functions.

\begin{lemma}\label{lemma1}
	Suppose that problem $\mathcal P_0$ is feasible and the optimal value $U_{\rm sen}^\star>0$. Then, every optimal transmit current of $\mathcal P_0$ lies in the electromagnetic-response subspace $\mathcal R$. Specifically, any optimal current can be represented as
	\begin{equation}
		\mathbf j_t^\star(\mathbf s)
		=
		\mathbf G(\mathbf s)\mathbf c_t^\star,
		\quad
		t=1,\ldots,T,
		\label{eq:optimal_current_subspace}
	\end{equation}
	for some finite-dimensional coefficient vectors
	$\mathbf c_t^\star\in\mathbb C^{3(K+Q)}$.
\end{lemma}

Before presenting the proof, we highlight the key idea behind Lemma~\ref{lemma1}. The transmit current $\mathbf j_t(\mathbf s)$ influences the communication and sensing performance only through the response kernels collected in $\mathbf G(\mathbf s)$. Any component orthogonal to the subspace generated by $\mathbf G(\mathbf s)$ is therefore invisible to the CI constraints and sensing fields, while still consuming transmit power. Hence, such a component is power-wasting, and an optimal transmit current must lie in the response subspace $\mathcal R$. The detailed proof is given below.

\noindent{\it Proof.}
	For any current distribution $\mathbf j_t(\mathbf s)$, decompose it as
	\begin{equation*}
		\mathbf j_t(\mathbf s)
		=
		\mathbf G(\mathbf s)\mathbf c_t
		+
		\boldsymbol\delta_t(\mathbf s),
	\end{equation*}
	where $\mathbf G(\mathbf s)\mathbf c_t\in\mathcal R$ and $\boldsymbol\delta_t(\mathbf s)\in\mathcal R^\perp$. By the definition of $\mathcal R^\perp$, 
	\begin{equation*}
		\int_{\mathcal S_T}
		\mathbf G^\herm(\mathbf s)
		\boldsymbol\delta_t(\mathbf s)
		d\mathbf s
		=
		\mathbf 0.
	\end{equation*}
	Equivalently,
	\begin{equation*}
		\int_{\mathcal S_T}
		\boldsymbol\Gamma_k(\mathbf s)
		\boldsymbol\delta_t(\mathbf s)
		d\mathbf s
		=
		\mathbf 0,\ \forall k,
		\quad
		\int_{\mathcal S_T}
		\boldsymbol\Gamma_q^{\rm s}(\mathbf s)
		\boldsymbol\delta_t(\mathbf s)
		d\mathbf s
		=
		\mathbf 0,\ \forall q.
	\end{equation*}
	Thus, $\boldsymbol\delta_t(\mathbf s)$ does not affect the communication received signals or sensing fields. Meanwhile, due to orthogonality,
	\begin{equation}
		\int_{\mathcal S_T}
		\|\mathbf j_t(\mathbf s)\|_2^2 d\mathbf s
		=
		\int_{\mathcal S_T}
		\|\mathbf G(\mathbf s)\mathbf c_t\|_2^2 d\mathbf s
		+
		\int_{\mathcal S_T}
		\|\boldsymbol\delta_t(\mathbf s)\|_2^2 d\mathbf s .
		\label{eq:power_decomposition}
	\end{equation}

	For an optimal solution $\{\mathbf j_t^\star(\mathbf s)\}_{t=1}^{T}$, if a nonzero orthogonal component $\boldsymbol\delta_t(\mathbf s)$ exists, removing it preserves all CI constraints and the sensing utility but strictly reduces the transmit power. Let $\bar{\mathbf j}_t(\mathbf s)$ denote the resulting response-subspace component. Then,
	\begin{equation*}
			\sum_{t=1}^{T}
			\int_{\mathcal S_T}
			\|\bar{\mathbf j}_t(\mathbf s)\|_2^2
			\,{\rm d}\mathbf s
			<
			\sum_{t=1}^{T}
			\int_{\mathcal S_T}
			\|\mathbf j_t^\star(\mathbf s)\|_2^2
			\,{\rm d}\mathbf s
			\leq
			TP_{\max}.
			\label{eq:strict_power_reduction}
	\end{equation*}
	Therefore, there exists $\alpha>1$ such that $\{\alpha\bar{\mathbf j}_t(\mathbf s)\}_{t=1}^{T}$ still satisfies the power constraint. Since the communication and sensing responses are linear in the current, the CI constraints remain feasible after this scaling, while the sensing utility becomes
	\begin{equation*}
		U_{\rm sen}\left(\{\alpha\bar{\mathbf j}_t\}\right)
		=
		\alpha^2 U_{\rm sen}^{\star}
		>
		U_{\rm sen}^{\star},
	\end{equation*}
	where the strict inequality follows from $U_{\rm sen}^{\star}>0$. This contradicts the optimality of $\{\mathbf j_t^\star(\mathbf s)\}_{t=1}^{T}$. Therefore, the orthogonal component of every optimal current must be zero, and every optimal current admits the representation in \eqref{eq:optimal_current_subspace}.
\hfill  $\blacksquare$

%
%
%
%

\vspace{-10pt}
	\subsection{Finite-Dimensional Reformulation of Problem ${\cal P}_0$}
	\label{subsec:basis_finite_reformulation}
With Lemma~\ref{lemma1}, we can convert problem ${\cal P}_0$ into a finite-dimensional problem. Specifically, define the response correlation matrix
	\begin{equation*}
		\mathbf C
		=
		\int_{\mathcal S_T}
		\mathbf G^\herm(\mathbf s)\mathbf G(\mathbf s)
		d\mathbf s
		\in\mathbb C^{3(K+Q)\times 3(K+Q)},
		\label{eq:response_correlation_matrix}
	\end{equation*}
	and let $\mathbf V_D$ and $\boldsymbol\Lambda_D=\operatorname{diag}(\lambda_1,\ldots,\lambda_D) \succ \bm 0$ be the matrices consisting of the eigenvectors and corresponding  positive eigenvalues of $\mathbf C$, respectively. Then, an orthonormal basis of $\mathcal R$ can be constructed as
	\begin{equation*}
		\boldsymbol\Xi(\mathbf s)
		=
		\mathbf G(\mathbf s)
		\mathbf V_D
		\boldsymbol\Lambda_D^{-1/2}
		=
		[
		\boldsymbol\xi_1(\mathbf s),
		\ldots,
		\boldsymbol\xi_D(\mathbf s)
		]
		\in\mathbb C^{3\times D},
		\label{eq:orthonormal_basis_matrix}
	\end{equation*}
which satisfies
$
		\int_{\mathcal S_T}
		\boldsymbol\Xi^\herm(\mathbf s)
		\boldsymbol\Xi(\mathbf s)
		d\mathbf s
		=
		\mathbf I_D .
$
	According to Lemma~\ref{lemma1}, the current distribution can be represented without loss of optimality as
	\begin{equation}
		\mathbf j_t(\mathbf s)
		=
		\boldsymbol\Xi(\mathbf s)\mathbf x_t,
		\quad
		\mathbf x_t\in\mathbb C^D,
		\quad
		t=1,\ldots,T,
		\label{eq:current_finite_expansion}
	\end{equation}
and the total transmit power reduces to
		\begin{equation}
		\sum_{t=1}^{T}
		\int_{\mathcal S_T}
		\|\mathbf j_t(\mathbf s)\|_2^2 d\mathbf s
		=
		\|\mathbf X\|_F^2,
		\label{eq:finite_power}
	\end{equation}
	where $\mathbf X=[\mathbf x_1,\ldots,\mathbf x_T]\in\mathbb C^{D\times T}$. By substituting~\eqref{eq:current_finite_expansion} and \eqref{eq:finite_power} into \eqref{eq:continuous_problem}, problem ${\cal P}_0$ can be reformulated as
	\begin{subequations}
		\label{eq:finite_problem}
		\begin{align}
			\mathcal P_1:
			\max_{\mathbf X,\{\boldsymbol\psi_k\}}
			&
			\sum_{t=1}^{T}
			\sum_{q=1}^{Q}
			\omega_q
			\|\mathbf A_q\mathbf x_t\|_2^2
			\label{eq:finite_problem_obj}
			\\
			\mathrm{s.t.}\quad
			&
			\Re \big\{
			\eta_m
			e^{-j\angle s_k[t]}
			\boldsymbol\psi_k^\herm
			\mathbf H_k\mathbf x_t
			\big\}
			\ge
			\beta_k,
			\, \forall m, k,t,
			\label{eq:finite_problem_ci}
			\\
			&
			\|\mathbf X\|_F^2
			\le
			TP_{\max},
			\label{eq:finite_problem_power}
			\\
			&
			\|\boldsymbol\psi_k\|_2=1,
			\quad
			\forall k,
			\label{eq:finite_problem_combiner}
		\end{align}
	\end{subequations}
where  $\small \mathbf A_q
=
\int_{\mathcal S_T}
\boldsymbol\Gamma_q^{\rm s}(\mathbf s)
\boldsymbol\Xi(\mathbf s)
\,\mathrm d\mathbf s
\in\mathbb C^{3\times D},~\forall~q$ and $\small \mathbf H_k
=
\int_{\mathcal S_T}
\boldsymbol\Gamma_k(\mathbf s)
\boldsymbol\Xi(\mathbf s)
\,\mathrm d\mathbf s
\in\mathbb C^{3\times D},~\forall~k$.

Problem \(\mathcal P_1\) involves only finite-dimensional variables and is therefore more tractable than the original functional problem \(\mathcal P_0\). Once a solution to \(\mathcal P_1\) is obtained, the corresponding continuous current distribution can be recovered directly from~\eqref{eq:current_finite_expansion}. We next develop an efficient solution for \(\mathcal P_1\).
\vspace{-10pt}
\subsection{Penalty Projected Gradient Algorithm}
	\label{subsec:penalty_algorithm}
	
	Although $\mathcal P_1$ is finite-dimensional, it is still nonconvex due to the maximized quadratic sensing objective, the bilinear coupling between $\mathbf x_t$ and $\boldsymbol\psi_k$, and the unit-norm constraints on the receive polarization combiners. To obtain a tractable solution, we develop a penalty projected gradient (PPG) algorithm.
	
	For the CI constraint of user $k$ at symbol time $t$, define the violation function as
	\begin{equation}
		e_{k,t,m}
		=
		\left[
		\beta_k
		-
		\Re\left\{
		\eta_m
		e^{-j\angle s_k[t]}
		\boldsymbol\psi_k^\herm
		\mathbf H_k\mathbf x_t
		\right\}
		\right]_+,
		\quad
		m=1,2,
		\label{eq:ci_violation}
	\end{equation}
	where $[a]_+=\max\{a,0\}$. With a penalty factor $\rho>0$, problem $\mathcal P_1$ is handled through the following penalized problem:
	\begin{equation}
		\label{eq:penalty_problem}
		\begin{aligned}
	\hspace{-10pt}		\mathcal P_2:\quad
			\min_{\mathbf X,\{\boldsymbol\psi_k\}}
			\quad
			&
			-
			\sum_{t=1}^{T}
			\mathbf x_t^\herm
			\mathbf R_{\rm s}
			\mathbf x_t
			+
			\frac{\rho}{2}
			\sum_{k=1}^{K}
			\sum_{t=1}^{T}
			\sum_{m=1}^{2}
			e_{k,t,m}^2
			\\
			\mathrm{s.t.}\quad
			& \eqref{eq:finite_problem_power} \, \text{and} \,  \eqref{eq:finite_problem_combiner} ~{\rm satisfied},
		\end{aligned}
	\end{equation}
	where $\mathbf R_{\rm s}=\sum_{q=1}^{Q}\omega_q\mathbf A_q^\herm\mathbf A_q$. The first term in $\mathcal P_2$ is the negative sensing utility and is handled
	as a nonconvex smooth term by projected gradient descent, while the second term penalizes the CI constraint violations. By increasing $\rho$, the CI constraints are progressively enforced, whereas the power and unit-norm constraints are handled explicitly by projection.
	
	Let $F_\rho(\mathbf X,\{\boldsymbol\psi_k\})$ denote the objective function of
	$\mathcal P_2$. The coefficient matrix is updated by
	\begin{equation}
		\widehat{\mathbf X}
		=
		\mathbf X
		-
		\mu_x
		\nabla_{\mathbf X}F_\rho,
		\label{eq:X_gradient_step}
	\end{equation}
	where $\mu_x>0$ is the stepsize. Then, the transmit-power constraint is enforced by the projection
	\begin{equation}
		\mathbf X^{+}
		=
		\min
		\left\{
		1, \
		\sqrt{
			{TP_{\max}}/{\|\widehat{\mathbf X}\|_F^2}
		}
		\right\}
		\widehat{\mathbf X}.
		\label{eq:X_projection}
	\end{equation}
	
	For each receive polarization combiner $\boldsymbol\psi_k$, the Euclidean gradient is projected onto the tangent space of the unit sphere as
	\begin{equation}
		\mathbf d_k
		=
		\nabla_{\boldsymbol\psi_k}F_\rho
		-
		\Re\left\{
		\boldsymbol\psi_k^\herm
		\nabla_{\boldsymbol\psi_k}F_\rho
		\right\}
		\boldsymbol\psi_k .
		\label{eq:psi_projected_gradient}
	\end{equation}
	Then, the combiner is updated and normalized by
	\begin{equation}
		\boldsymbol\psi_k^{+}
		=
		\frac{
			\boldsymbol\psi_k
			-
			\mu_\psi
			\mathbf d_k
		}
		{
			\left\|
			\boldsymbol\psi_k
			-
			\mu_\psi
			\mathbf d_k
			\right\|_2
		},
		\quad
		k=1,\ldots,K,
		\label{eq:psi_update}
	\end{equation}
	where $\mu_\psi>0$ is the stepsize.
	
	The overall subspace-based penalty projected gradient algorithm is summarized in Algorithm~\ref{alg:penalty_projected}.
	\begin{algorithm}[t]
		\caption{Subspace-Based PPG Algorithm}
		\label{alg:penalty_projected}
		\begin{algorithmic}[1]
			\State Construct $\boldsymbol\Xi(\mathbf s)$, $\{\mathbf A_q\}_{q=1}^{Q}$,
			and $\{\mathbf H_k\}_{k=1}^{K}$.
			\State Initialize $\mathbf X$ with $\|\mathbf X\|_F^2\le TP_{\max}$ and
			$\{\boldsymbol\psi_k:\|\boldsymbol\psi_k\|_2=1\}$.
			\State Set $\rho_{\max}$, $\rho=\rho_0$ and choose $\kappa>1$.
			\Repeat
			\Repeat
			\State Compute the CI violations $\{e_{k,t,m}\}$ by
			\eqref{eq:ci_violation}.
			\State Update $\mathbf X$ by \eqref{eq:X_gradient_step} and
			\eqref{eq:X_projection}.
			\State Update each $\boldsymbol\psi_k$ by
			\eqref{eq:psi_projected_gradient} and \eqref{eq:psi_update}.
			\Until{the penalized objective change is below $\epsilon_{\rm obj}$}
			\State Update $\rho\leftarrow \min\{\kappa\rho,\rho_{\max}\}$.
			\Until{the maximum CI violation is below $\epsilon_{\rm CI}$}
			\State Recover $\mathbf j_t(\mathbf s)=\boldsymbol\Xi(\mathbf s)\mathbf x_t$,
			$t=1,\ldots,T$.
		\end{algorithmic}
	\end{algorithm}
\vspace{-10pt}
	\subsection{Complexity Analysis}	
		The per-iteration complexity of Algorithm~\ref{alg:penalty_projected} is dominated by the updates of $\mathbf X$ and $\{\boldsymbol\psi_k\}$. Since $\mathbf R_{\rm s}=\sum_{q=1}^{Q}\omega_q\mathbf A_q^\herm\mathbf A_q$ is fixed, it is computed once before the iterations. Each inner iteration requires $\mathcal O(D^2T)$ operations for the sensing-gradient term $\mathbf R_{\rm s}\mathbf X$ and $\mathcal O(KTD)$ operations for the CI-penalty terms. Therefore, with $I_{\rm in}$ inner iterations and $I_{\rm out}$ outer penalty updates, the overall iterative complexity is
		\begin{equation}
			\mathcal O\!\left(I_{\rm out}I_{\rm in}(D^2T+KTD)\right).
		\end{equation}
		Since $D\le 3(K+Q)$, increasing the number of users or sensing targets enlarges the response subspace and increases the computational cost.

\section{Simulation Results}
\label{sec:simulation_results}

We evaluate the proposed CAPA-aided ISAC design through Monte Carlo simulations with $1,000$ independent trials. Unless otherwise specified, the carrier frequency is $f_c=2.4$ GHz, the square aperture size is $L_x=L_y=0.6~{\rm m}$, and the maximum transmit power is $P_{\max}=5$. The system serves $K=2$ users and illuminates $Q=2$ targets over a block of $T=4$ symbols with $8$-PSK modulation. The CI margin and sensing weights are set to $\beta_k=0.05$ and $\omega_q=10$, respectively. The users are located at $(\theta_k,\phi_k,R_k)=(-40^\circ,20^\circ,20{\rm m})$ and $(40^\circ,20^\circ,20{\rm m})$, while the sensing targets are located at $(\theta_q,\phi_q)=(-45^\circ,45^\circ)$ and $(45^\circ,45^\circ)$. Distance-dependent path loss is included in the communication channels.
For the proposed PPG algorithm, the penalty parameter is initialized as $\rho_0=3\times10^2$ and increased by a factor of $1.008$ until $\rho_{\max}=8\times10^3$. The stepsizes are $\mu_x=\mu_\psi=0.02$, the stopping tolerances are $\epsilon_{\rm obj}=\epsilon_{\rm CI}=10^{-4}$, and the maximum number of iterations is 800. The matrix $\mathbf X$ is randomly initialized within the transmit-power ball, and each $\boldsymbol\psi_k$ is randomly generated and normalized to unit norm. Multiple random initializations are adopted to improve robustness.

We compare the proposed method with Fourier-CAPA~\cite{zhang2023pattern} and Digital-SPDA~\cite{sanguinetti2022wavenumber}. Fourier-CAPA expands the continuous current over truncated two-dimensional Fourier bases with $N_x=N_y=5$, and $N_F=121$ scalar bases. Digital-SPDA discretizes the aperture into a half-wavelength-spaced fully digital array. Fixed-polarization variants with $\boldsymbol\psi_{\rm fix}=[1,0,0]^{\mathsf T}$ are also included.

Fig.~\ref{fig:illumination_map} shows the normalized sensing illumination maps over the azimuth--elevation plane. The proposed subspace-CAPA scheme forms focused beams toward the sensing targets and communication users. Compared with the fixed-polarization counterpart, it achieves stronger target illumination, since the optimized polarization combiners enhance the effective communication channels and preserve more transmit-current freedom for sensing-field shaping. The weaker and less focused illumination of the digital-SPDA baselines further highlights the advantage of continuous-aperture current optimization.

Fig.~\ref{fig:performance}(a) compares the average sensing utility versus transmit power. The proposed subspace-CAPA scheme consistently achieves the highest sensing utility, with a more pronounced gain at larger transmit powers. Fourier-CAPA yields lower utility due to its truncated basis representation, while digital-SPDA performs the worst, confirming the benefit of continuous-aperture current optimization.

Fig.~\ref{fig:performance}(b) shows the BER performance versus receive SNR. The proposed subspace-CAPA scheme achieves the lowest BER and the fastest error-rate decay. The degradation of its fixed-polarization variant confirms the benefit of polarization optimization, while the higher BERs of Fourier-CAPA and digital-SPDA highlight the advantage of response-subspace current optimization. This advantage is not restricted to the specific setting in Fig.~\ref{fig:performance}; similar trends are expected under broader system configurations, since the proposed design more effectively exploits the relevant electromagnetic degrees of freedom.

	\begin{figure}[t]
		\centering
		\includegraphics[width=0.9\linewidth, trim=10 10 10 30,clip]{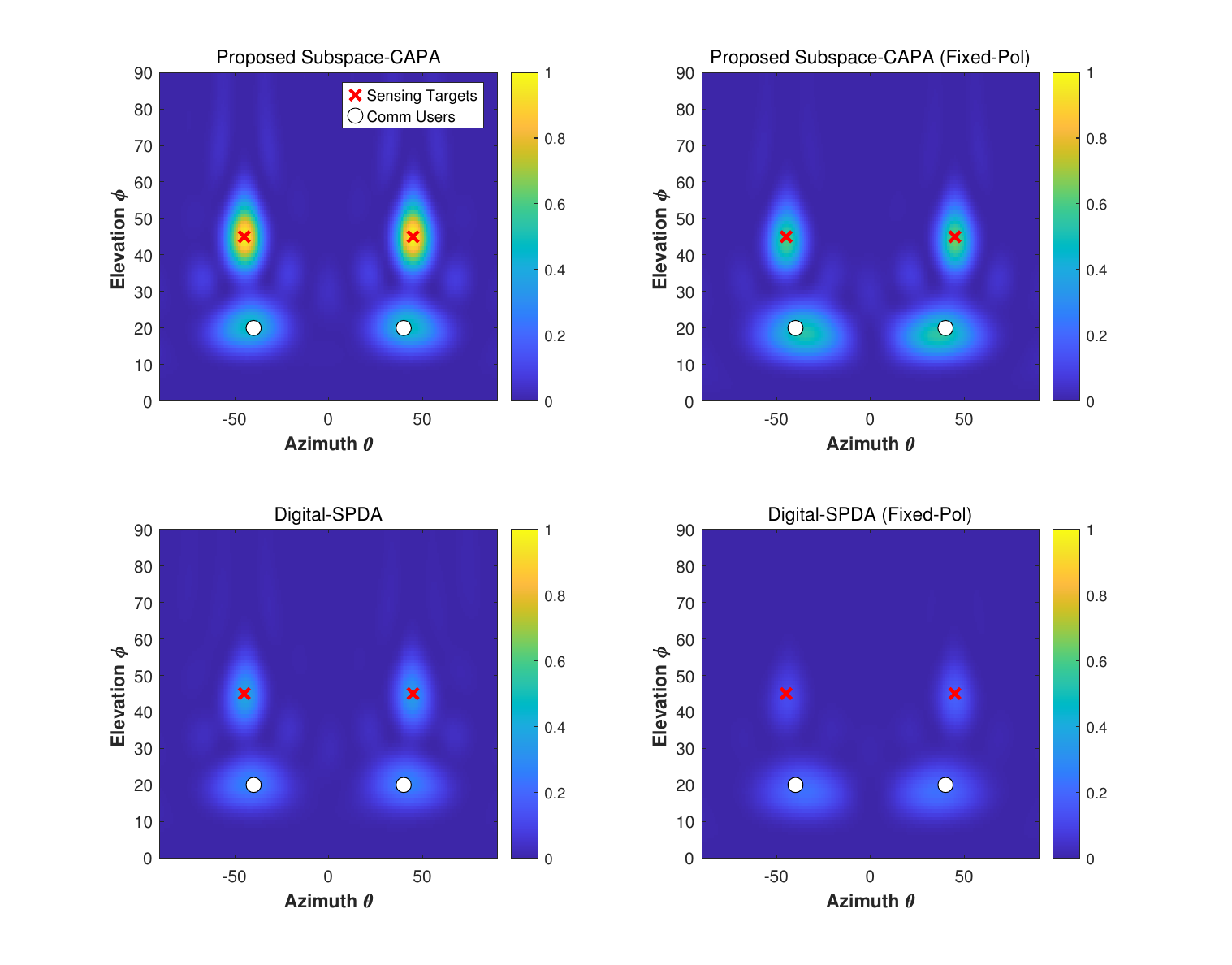}
		\caption{Sensing illumination maps.}
		\label{fig:illumination_map}
	\end{figure}
	
\begin{figure}[t]
	\centering
	\subfloat[Sensing utility]{
		\includegraphics[width=0.48\linewidth, trim=10 10 10 25,clip]{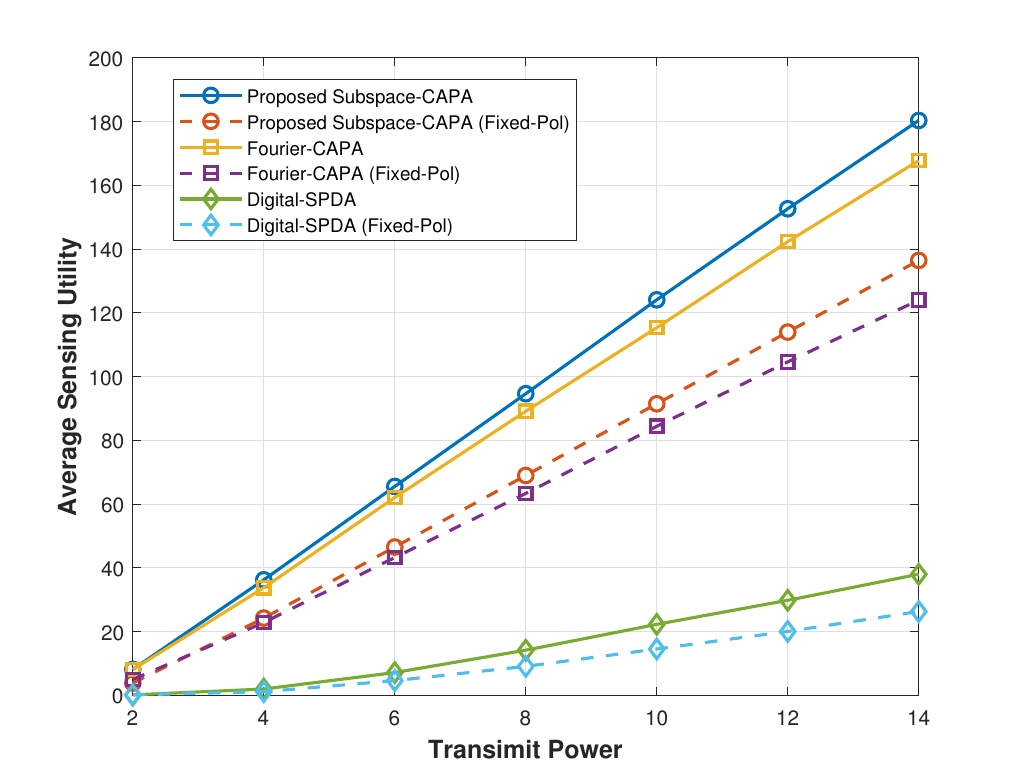}
		\label{fig:sensing_power}
	}
	\hspace{-6mm}
	\subfloat[BER ]{
		\includegraphics[width=0.48\linewidth, trim=10 9 10 22,clip]{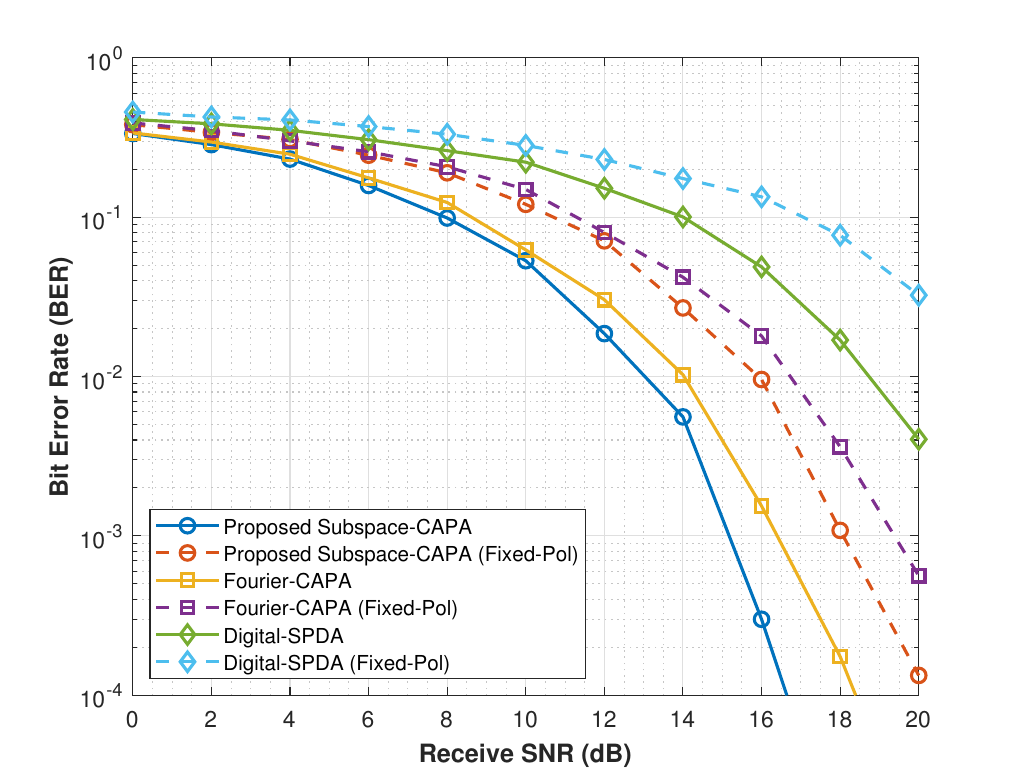}
		\label{fig:ber_snr}
	}
	\caption{Performance comparison.}
	\label{fig:performance}
\end{figure}

\section{Conclusion}
In this paper, we investigated a CAPA-enabled downlink ISAC system with SLP and receive polarization combining. By exploiting the electromagnetic-response structure, we reduced the infinite-dimensional current design to a finite-dimensional vector optimization problem and developed a tailored penalty projected-gradient algorithm. Simulation results showed that the proposed method outperforms conventional Fourier-basis and SPDA baselines in both sensing utility and BER performance.	
	
\vspace{-10pt}	
	\bibliographystyle{IEEEtran}
	\bibliography{ref}

@article{xiu2025movable,
	title={{Movable antenna enabled ISAC beamforming design for low-altitude airborne vehicles}},
	author={Xiu, Yue and Yang, Songjie and Lyu, Wanting and Yeoh, Phee Lep and Li, Yonghui and Ai, Yi},
	journal={IEEE Wireless Commun. Lett.},
	volume={14},
	number={5},
	pages={1311--1315},
	year={2025},
	publisher={IEEE}
}

@article{wang2026movable,
	title={{Movable Antenna Enhanced MIMO Communications with Spatial Modulation}},
	author={Wang, Kaihe and Yang, Ran and Zhu, Lipeng and Xi, Rongyan and Xiu, Yue and Zhang, Zhongpei},
	journal={arXiv preprint arXiv:2601.12788},
	year={2026}
}

@article{yang2026movable,
	title={{Movable Antenna Enhanced Dual-Functional Radar-Communication: A Symbol-Level Precoding Approach}},
	author={Yang, Ran and Wei, Ning and Dong, Zheng and Assi, Chadi and Li, You and Xu, Fei and Xiu, Yue},
	journal={arXiv preprint arXiv:2605.27839},
	year={2026}
}

@article{yang2025secure,
	title={{Secure transceiver design for discrete RIS enhanced dual-functional radar-communication: A symbol-level precoding approach}},
	author={Yang, Ran and Dong, Zheng and Xiu, Yue and Wei, Ning and Li, Yan and Li, Yonghui},
	journal={IEEE Wireless Commun. Lett.},
	volume={14},
	number={4},
	pages={1034--1038},
	year={2025},
	publisher={IEEE}
}

@ARTICLE{zhang2023pattern,
	author={Zhang, Zijian and Dai, Linglong},
	journal={IEEE J. Sel. Areas Commun.}, 
	title={{Pattern-Division Multiplexing for Multi-User Continuous-Aperture MIMO}}, 
	year={2023},
	volume={41},
	number={8},
	pages={2350-2366},
	doi={10.1109/JSAC.2023.3288244}}

@ARTICLE{wang2025beamforming,
	author={Wang, Zhaolin and Ouyang, Chongjun and Liu, Yuanwei},
	journal={IEEE Trans. Wireless Commun.}, 
	title={{Optimal Beamforming for Multi-User Continuous Aperture Array (CAPA) Systems}}, 
	year={2025},
	volume={73},
	number={10},
	pages={9207-9221},
	doi={10.1109/TCOMM.2025.3554644}}

@article{si2025joint,
	title={{Joint DOA and Attitude Sensing Based on Tri-Polarized Continuous Aperture Array}},
	author={Si, Haonan and Wang, Zhaolin and Guo, Xiansheng and Zhang, Jin and Liu, Yuanwei},
	journal={arXiv preprint arXiv:2510.02029},
	year={2025}
}

@ARTICLE{jiang2025cramer,
	author={Jiang, Hao and Wang, Zhaolin and Liu, Yuanwei and Nallanathan, Arumugam},
	journal={IEEE Trans. Wireless Commun.}, 
	title={Cramér–Rao Bound Optimization for Near-Field Sensing With Continuous-Aperture Arrays}, 
	year={2026},
	volume={25},
	number={},
	pages={7032-7047},
	doi={10.1109/TWC.2025.3628549}}

@ARTICLE{zhang2025exploiting,
	author={Zhang, Yue and Ouyang, Chongjun and Shan, Hangguan and Liu, Yuanwei and Zhou, Yong and Shi, Zhiguo},
	journal={IEEE Trans. Wireless Commun.}, 
	title={{Exploiting Continuous-Aperture Arrays in Integrated Sensing and Communication Systems}}, 
	year={2025},
	volume={24},
	number={11},
	pages={9539-9555},
	doi={10.1109/TWC.2025.3573872}}

@article{zhao2026continuous,
	title={{Continuous-Aperture Array-Based ISAC Over Fading Channels}},
	author={Zhao, Boqun and Ouyang, Chongjun and Zhang, Xingqi and Liu, Yuanwei},
	journal={arXiv preprint arXiv:2603.03184},
	year={2026}
}

@ARTICLE{ye2026optimal,
	author={Ye, Junjie and Wang, Zhaolin and Liu, Yuanwei and Zhang, Peichang and Huang, Lei and Nallanathan, Arumugam},
	journal={IEEE Trans. Wireless Commun.}, 
	title={{Optimal Waveform Design for Continuous Aperture Array (CAPA)-Aided ISAC Systems}}, 
	year={2026},
	volume={25},
	number={},
	pages={19082-19098},
	doi={10.1109/TWC.2026.3702839}}

@article{masouros2015exploiting,
	title={{Exploiting known interference as green signal power for downlink beamforming optimization}},
	author={Masouros, Christos and Zheng, Gan},
	journal = {IEEE Trans. Signal Process.},
	volume={63},
	number={14},
	pages={3628--3640},
	year={2015},
	publisher={IEEE}
}

@article{sanguinetti2022wavenumber,
	title={{Wavenumber-division multiplexing in line-of-sight holographic MIMO communications}},
	author={Sanguinetti, Luca and D’Amico, Antonio Alberto and Debbah, Merouane},
	journal = {IEEE Trans. Wireless Commun.},
	volume={22},
	number={4},
	pages={2186--2201},
	year={2022},
	publisher={IEEE}
}
\end{document}